\documentclass[letterpaper,11pt]{article} 
\usepackage[utf8]{inputenc}

\usepackage{amsmath,amssymb,amstext,amsthm}
\usepackage{algpseudocode,algorithm}
\usepackage[in]{fullpage} 
\usepackage{booktabs}
\usepackage{xspace}
\usepackage{pgfplots}
\pgfplotsset{compat=1.18}
\usepackage{authblk}

\newtheorem{theorem}{Theorem}

\newtheorem{lemma}{Lemma}
\newtheorem{corollary}{Corollary}

\newtheorem{definition}{Definition}

\usepackage{todonotes}

\usepackage{hyperref}

\newcommand{\cF}{\ensuremath{\mathcal{F}}\xspace}
\newcommand{\cFd}{\ensuremath{\mathcal{F}\downarrow}\xspace}
\newcommand{\bitsize}{N}
\newcommand{\set}[1]{\left\{#1\right\}}

\title{Quantum Algorithms for Graph Coloring and other Partitioning, Covering, and Packing Problems}
\author[1]{Serge Gaspers}
\author[2]{Jerry Zirui Li} 
\affil[1]{UNSW Sydney, Australia, \href{mailto:serge.gaspers@unsw.edu.au}{serge.gaspers@unsw.edu.au}}
\affil[2]{James Ruse Agricultural High School. Sydney, Australia, \href{mailto:zirui.li7@gmail.com}{zirui.li7@gmail.com}}

\date{}

\begin{document}

\maketitle

\begin{abstract}
    Let $U$ be a universe on $n$ elements, let $k$ be a positive integer, and let $\cF$ be a family of (implicitly defined) subsets of $U$.
    We consider the problems of partitioning $U$ into $k$ sets from $\cF$, covering $U$ with $k$ sets from $\cF$, and packing $k$ non-intersecting sets from $\cF$ into $U$.
    Classically, these problems can be solved via inclusion--exclusion in $2^n n^{O(1)}$ time \cite{BjorklundHK09}.
    Quantumly, there are faster algorithms for graph coloring with running time $O(1.9140^n)$ \cite{ShimizuM22} and for Set Cover with a small number of sets with running time $O\left(1.7274^n |\cF|^{O(1)}\right)$ \cite{AmbainisBIKPV19}.
    In this paper, we give a quantum speedup for Set Partition, Set Cover, and Set Packing whenever there is a classical enumeration algorithm that lends itself to a quadratic quantum speedup, which, for any subinstance on a set $X\subseteq U$, enumerates at least one member of a $k$-partition, $k$-cover, or $k$-packing (if one exists) restricted to (or projected onto, in the case of $k$-cover) the set $X$ in $c^{|X|} n^{O(1)}$ time with $c<2$.
    Our bounded-error quantum algorithm runs in time $(2+c)^{n/2} n^{O(1)}$ for Set Partition, Set Cover, and Set Packing. It is obtained by combining three algorithms that have the best running time for some values of $c$. When $c\le 1.147899$, our algorithm is slightly faster than $(2+c)^{n/2} n^{O(1)}$; when $c$ approaches $1$, it matches the $O\left(1.7274^n |\cF|^{O(1)}\right)$ running time of \cite{AmbainisBIKPV19} for Set Cover when $|\cF|$ is subexponential in $n$.
    
    For covering, packing, and partitioning into maximal independent sets, maximal cliques, maximal bicliques, maximal cluster graphs, maximal triangle-free graphs, maximal cographs, maximal claw-free graphs, maximal trivially-perfect graphs, maximal threshold graphs, maximal split graphs, maximal line graphs, and maximal induced forests, we obtain bounded-error quantum algorithms with running times ranging from $O(1.8554^n)$ to $O(1.9629^n)$.
    Packing and covering by maximal induced matchings can be done quantumly in $O(1.8934^n)$ time.
    
    For Graph Coloring (covering with $k$ maximal independent sets), we further improve the running time to $O(1.7956^n)$ by leveraging faster algorithms for coloring with a small number of colors to better balance our divide-and-conquer steps.
    For Domatic Number (packing $k$ minimal dominating sets), we obtain a $O((2-\varepsilon)^n)$ running time for some $\varepsilon > 0$.
    
    Several of our results should be of interest to proponents of classical computing:
    \begin{itemize}
        \item We present an inclusion-exclusion algorithm with running time $O^*\left(\sum_{i=0}^{\lfloor \alpha n \rfloor} \binom{n}{i}\right)$, which determines, for each $X\subseteq U$ of size at most $\alpha n$, $0\le \alpha \le 1$, whether $(X,\cF)$ has a $k$-cover, $k$-partition, or $k$-packing. This running time is best-possible, up to polynomial factors.
        \item We prove that for any linear-sized vertex subset $X\subseteq V$ of a graph $G=(V,E)$, the number of minimal dominating sets of $G$ that are subsets of $X$ is $O((2-\varepsilon)^{|X|})$ for some $\varepsilon > 0$.
    \end{itemize}
\end{abstract}

\newpage

\section{Introduction}

Graph Coloring is an example of a problem requiring to partition an $n$-element set $U$ into $k$ sets from a family $\cF$. In this case $U$ is the vertex set of a graph $G$ and $\cF$ is implicitly defined as the independent sets of $G$. We can also view Graph Coloring as a covering problem where the vertex set needs to be covered with $k$ maximal independent sets.

In 2006, \cite{BjorklundH06} and \cite{Koivisto06} independently solved Graph Coloring in $O^*(2^n)$ time via a new inclusion--exclusion approach, along with other partitioning and covering problems. The approach has been used for packing problems as well, and has been generalized to solve more generic subset convolution problems \cite{BjorklundHKK07,BjorklundHKK11,BjorklundHK09}.

In this work, we give faster quantum algorithms for a range of partitioning, covering, and packing problems, including Graph Coloring and Domatic Number. To do this, we use the framework of \cite{AmbainisBIKPV19} where a preprocessing step computes solutions to small subinstances and stores them in QRAM. These solutions are then accessed by a divide-and-conquer algorithm which enjoys a quadratic speedup in quantum models of computation via techniques such as Grover's search \cite{Grover96}. For the preprocessing step (\autoref{sec:pre}), we adapt the afore-mentioned inclusion-exclusion approach \cite{BjorklundHK09} to compute the solutions to all subinstances induced by a small subset of $U$ up to roughly $n/4$ elements. For the divide-and-conquer step (\autoref{sec:divide}) one would ideally like to divide $U$ into two halves; unfortunately, optimal partitions\footnote{For conciseness, we will mainly discuss partitions in the introduction. The treatment of covers and packings is similar.} may not allow for such a balanced split. However, we can restrict the divide-and-conquer step to divide $U$ into two parts where in the larger part (equivalently, in both parts) the removal of one set of an optimal partition results in at most $n/2$ elements. Finding one set of the optimal partition is done via an algorithm that enumerates all relevant candidate sets; for Graph Coloring, it enumerates the maximal independent sets in the graph induced by the subset $X \subseteq U$ under consideration. Importantly, we need that this enumeration can be done in $O^*\left(c^{|X|}\right)$ time, for some $c<2$ by a classical algorithm that has a quadratic quantum speedup, so that after two levels of divide-and-conquer, the overall quadratic quantum speedup outperforms the classical $O^*(2^n)$ running time.

Our algorithm differs from the previously fastest quantum algorithm for Graph Coloring \cite{ShimizuM22} in both the preprocessing step and the divide-and-conquer step. Our preprocessing step is deterministic and its running time is optimal, matching the size of the output up to polynomial factors; the preprocessing step of \cite{ShimizuM22} is a bounded-error quantum algorithm whose running time is a multiplicative factor of $3^{|X|/6}$ slower than ours for each small (up to size roughly $n/4$) subset $X\subseteq U$. For the divide-and-conquer step, our divide-then-enumerate strategy is described above; \cite{ShimizuM22} employ an enumerate-then-divide strategy, where the enumeration is done on the set to be divided and the remainder is then divided into two sets of size at most half the original set. It turns out that blending the the divide-then-enumerate and the enumerate-then-divide strategy gives faster algorithms when $c \le 1.147899$ (\autoref{sec:smallc}). For $c\le 1.0872$ case, we also use a third level of divide-and-conquer, and when $c$ approaches $1$, our $O(1.7274^n)$ running time matches the running time for Set Cover with a subexponential number of sets of \cite{AmbainisBIKPV19}.

This gives improved algorithms for a range of partitioning, covering, and packing problems (\autoref{sec:applications}). We further improve the running time for Graph Coloring (\autoref{sec:chromatic}) to $O(1.7956^n)$ by leveraging faster algorithms for a small number of colors \cite{ShimizuM22}. Our algorithm considers large subsets of vertices ($\ge 0.48 n$) an checks whether they are $5$-colorable, $6$-colorable with no large $5$-colorable subset, and in some cases $7$-colorable via a new $7$-Coloring algorithm that relies on the preprocessing step. The advantage of excluding such cases from further consideration is that we can make the divide-and-conquer steps more balanced.

For Domatic Number, at first glance it seems that our approach cannot be used. The issue is that when considering a vertex subset $X$, even though there is an algorithm that enumerates all minimal dominating sets of $G[X]$ in $O(1.7159^n)$ time \cite{FominGPS08}, this is insufficient for our purposes: we need to enumerate minimal vertex subsets of $X$ that dominate all of $G$, not just $G[X]$, in $O^*\left(c^{|X|}\right)$ time for some $c<2$. In \autoref{sec:dom}, we show that such an enumeration algorithm (with a quadratic quantum speedup) indeed exists provided that $|X|=\Omega(n)$. This then also gives a bounded-error quantum algorithm for Domatic Number running in $O((2-\varepsilon)^n)$ time.







\section{Preliminaries}

For a proposition $P$, the \emph{Iverson bracket} $[P]$ is a function that returns $1$ if $P$ is true and $0$ otherwise.

\paragraph{Asymptotic notation.}
The $O^*$-notation is similar to the usual $O$-notation but allows to hide polynomial factors in the input size. The $\Tilde{O}$-notation hides polylogarithmic factors.
We make heavy use of Stirling's approximation for factorials, which implies that $\binom{n}{k} = O^*\left(\left(\frac{n}{k}\right)^k \cdot \left( \frac{n}{n-k}\right)^{n-k}\right)$, and of the binomial theorem, $\sum_{k=0}^n \binom{n}{k} x^k y^{n-k} = (x+y)^n$.

\paragraph{Set Systems.}
An {\em implicit set system} \cite{FominGLS19} is a function $\Phi$ that takes as input a string $I \in \{0,1\}^*$ and outputs a set system $(U_I, \cF_I)$, where $U_I$ is a universe and $\cF_I$ is a collection of subsets of $U_I$. The string $I$ is referred to as an \emph{instance} and we denote by $|U_I| = n$ the size of the universe and by $|I|=\bitsize$ the size of the instance. 
We assume that $\bitsize\ge n$.
The implicit set system $\Phi$ is said to be \emph{polynomial time computable} if (a) there exists a polynomial time algorithm that given $I$, produces $U_I$, and (b) there exists a polynomial time algorithm that given $I$, $U_I$ and a subset $S$ of $U_I$, determines whether $S \in \cF_I$.
Throughout this paper, we consider only polynomial time computable implicit set systems.
We define a \emph{subset polynomial time computable} implicit set system $\Phi$ to be a polynomial time computable set system, where (c) there exists a polynomial time algorithm that given $I$, $U_I$ and a subset $S$ of $U_I$, determines whether $S \subseteq S'$ for some $S'\in \cF_I$. This is equivalent to determining whether $S\in \cFd$, where the downward closure $\cFd$ of $\cF$ contains all sets in $\cF$ and their subsets.

For any subset of elements $X\subseteq U$, an ordered tuple $(S_1,\dots,S_k)$ of $k$ sets from \cF is a \emph{$k$-cover} for $X$ if the union of these sets is $X$; it is a \emph{$k$-packing} for $X$ if the $S_i$'s are contained in $X$ and are pairwise non-intersecting; it is a \emph{$k$-partition} for $X$ if it is both a $k$-cover and a $k$-packing for $X$.

For a subset polynomial time computable implicit set system $\Phi$, the input of the $\Phi$-Set Cover problem is an instance $I$ and an integer $k$, and the question is whether the set system $\Phi(I)=(U_I,\cF_I)$ has a $k$-cover. This is equivalent to asking whether $(U_I,\cF_I\downarrow)$ has a $k$-cover and therefore, we assume that $\cF_I = \cF_I\downarrow$ whenever discussing $k$-covers. For a polynomial time computable implicit set system $\Phi$, the input of the $\Phi$-Set Partition and $\Phi$-Set Packing problem is an instance $I$ and an integer $k$, and the question is whether the set system $\Phi(I)=(U_I,\cF_I)$ has a $k$-partition or $k$-packing, respectively. We generally omit $\Phi$ and the subscript $I$ when they are clear from context.

\paragraph{Graphs.}
In a graph $G=(V,E)$, the \emph{open neighborhood} of a vertex $v$, denoted $N_G(v)$ is the set containing all vertices adjacent to $v$ in $G$.
The \emph{closed neighborhood} of $v$ in $G$ also contains $v$ itself, and is denoted $N_G[v] = \{v\} \cup N_G(v)$. Again, we may omit the subscript $G$.
For a vertex subset $X\subseteq V$, the graph $G-X$ is obtained from $G$ by removing the vertices in $X$ and all their incident edges; the graph $G[X]$ induced on $X$ is the graph $G-(V\setminus X)$.

\paragraph{Quantum Algorithms.}
It is known that most classical branching algorithms have a quadratic speedup on quantum machines. As \cite{ShimizuM22}, we also rely on the following results.

\begin{theorem}[\cite{Grover96,BoyerBHT98}]
    Let $A : \{1, 2, \dots , N \} \rightarrow \{0, 1\}$ be a bounded-error quantum algorithm with running time $T$. Then, there is a bounded-error quantum algorithm computing $\bigvee_{x\in\{1,\dots,N\}} A(x)$ with running time $\Tilde{O}(\sqrt{N}T)$.
\end{theorem}
\begin{theorem}[\cite{DurrH96}]
    Let $A : \{1, 2, \dots , N \} \rightarrow \{0, 1\}$ be a bounded-error quantum algorithm with running time $T$. Then, there is a bounded-error quantum algorithm computing $\min_{x\in\{1,\dots,N\}} A(x)$ with running time $\Tilde{O}(\sqrt{N}T)$.
\end{theorem}
In our context, $A$ is an algorithm exploring paths in superposition from the root to the leaves of the search tree of a classical branching algorithm.
The amplitudes of this exploration depend on estimates of the sizes of the subtrees, either by relying on an analysis of the classical branching algorithm \cite{Furer08,ShimizuM22} or by on-the-fly estimations \cite{AmbainisK17}.
We speak of a \emph{simple} branching algorithm when the exploration of one root-to-leaf path is independent of the other paths; this excludes, for example, algorithms relying on clause learning, re-use of computation done in earlier branches, and branch-and-bound.
For a simple branching algorithm with running time $O^*(c^n)$, one obtains a bounded-error quantum algorithm with running time $O^*(c^{n/2})$ in this way; we simply say that we apply Grover's search to the branching algorithm.

\section{Preprocessing small subsets}
\label{sec:pre}

For $\alpha\in [0,1]$, a subset $X$ of $U$ is \emph{$\alpha$-small} if $|X|\le \alpha n$.
Denote by $s(n,\alpha) = \sum_{i=0}^{\lfloor \alpha n \rfloor} \binom{n}{i}$ the number of $\alpha$-small subsets of $U$.
In this section, we consider the problem of counting the number of $k$-covers, $k$-packings, and $k$-partitions for each $\alpha$-small subset of $U$.
When considering $k$-covers, we assume that $\cF$ has been replaced by $\cFd$. This is because when we would like to cover a subset of elements of $X$, we may use a set from $\cF$ that also contains elements outside of $X$. Since $\Phi$ is subset polynomial time computable in the $\Phi$-Set Cover problem, we may as well replace $\cF$ by $\cFd$; this makes the discussion of covering, partitioning, and packing problems more uniform.
Our algorithms run in $O^*\left(s(n,\alpha) \right)$ time, which is best possible, since the output is a list of $s(n,\alpha)$ integers.

This section heavily relies on previous $O^*(2^n)$ inclusion-exclusion approaches \cite{BjorklundHKK11,BjorklundHK09} to compute the number of $k$-covers, $k$-packings, and $k$-partitions for $U$ and these results are well-known when $\alpha=1$.

We start by defining the $\alpha$-small zeta transform, which is central to this section.

\begin{definition}\label{def:small-zeta}
    Let $f$ be a function from subsets of the universe $U$ to an algebraic ring $R$.
    The \emph{$\alpha$-small zeta transform} of $f$, denoted $f\zeta_\alpha$ is
    \begin{align*}
        f\zeta_\alpha(X) = \sum_{Y\subseteq X} f(Y)
    \end{align*}
    for any $\alpha$-small $X\subseteq U$.
\end{definition}
The $1$-small zeta transform is also called the \emph{zeta transform} and the $\alpha$-small zeta transform is precisely the restriction of the zeta transform to $\alpha$-small sets.
Throughout this paper we assume that arithmetic operations in the ring $R$ take $O^*(1)$ time and each ring element is represented using $O^*(1)$ space.

\begin{definition}\label{def:small-mobius}
    Let $f$ be a function from subsets of the universe $U$ to an algebraic ring $R$.
    The \emph{$\alpha$-small M\"{o}bius transform} of $f$, denoted $f\mu_\alpha$ is
    \begin{align*}
        f\mu_\alpha(X) = \sum_{Y\subseteq X} (-1)^{|X\setminus Y|} f(Y)
    \end{align*}
    for any $\alpha$-small $X\subseteq U$.
\end{definition}

It is well-known (see, e.g., \cite{FominK10}) that $f \zeta_\alpha \mu_\alpha = f \mu_\alpha \zeta_\alpha = f$ when $\alpha = 1$, and the same is true when $\alpha \ne 1$.

\begin{lemma}
    The $\alpha$-small zeta transform $f\zeta_\alpha$ and the $\alpha$-small M\"{o}bius transform $f\mu_\alpha$ can be computed in $O^*(s(n,\alpha))$ time.
\end{lemma}
\begin{proof}
    We start with $f\zeta_\alpha$ and proceed as in Yates's method \cite{Yates37}.
    Consider an arbitrary ordering of the elements of $U = \{v_1, \dots, v_n\}$.
    The algorithm considers each $\alpha$-small $X\subseteq U$ by increasing order of cardinality.
    
    Set $g_0(X) = f(X)$.
    Then, iterate over the elements of $U$ in the ordering fixed above.
    When processing element $v_i$, set
    \begin{align*}
        g_i(X) &= g_{i-1}(X) + [v_i \in X] \cdot g_{i-1}(X\setminus \{v_i\}).
    \end{align*}
    Finally, set $f\zeta_\alpha(X) = g_n(X)$.
    
    Correctness can be shown by induction on $i$ by observing that
    \begin{align*}
        g_i(X) = \sum_{\{v_{i+1},\dots,v_n\}\cap X \subseteq Y \subseteq X} f(Y).
    \end{align*}
    
    For each set $X$ the computation takes $O^*(1)$ time, and the number of sets $X$ to be considered is $s(n,\alpha)$.
    
    To compute $f\mu_\alpha$, we use the fact that $\mu_\alpha = \sigma_\alpha \zeta_\alpha \sigma_\alpha$, where the $\alpha$-small odd-negation transform is
    \begin{align*}
        f\sigma_\alpha(X) = (-1)^{|X|} f(X),
    \end{align*}
    defined for any $\alpha$-small $X\subseteq U$.
    Indeed,
    \begin{align*}
        f\sigma_\alpha \zeta_\alpha \sigma_\alpha(X) &= (-1)^{|X|} \cdot \sum_{Y\subseteq X} (-1)^{|Y|} \cdot f(Y)\\
        & = \sum_{Y\subseteq X} (-1)^{|X|+|Y|} f(Y)\\
        & = \sum_{Y\subseteq X} (-1)^{|X\setminus Y|} f(Y)
    \end{align*}
    since $|X|+|Y|$ and $|X|-|Y| = |X\setminus Y|$ have the same parity.
    The result now follows, because, for a function $g$ and a set $X$, $g\sigma_\alpha(X)$ can be computed in $O^*(1)$ time.
\end{proof}
We refer to these algorithms as the \emph{fast $\alpha$-small zeta transform} and the \emph{fast $\alpha$-small M\"{o}bius transform}.

By inclusion-exclusion, the number of $k$-covers for a subset $X\subseteq U$ is
\begin{align}\label{eq:cover}
    c_k(\cF,X) = \sum_{Y\subseteq X} (-1)^{|X\setminus Y|} a(Y)^k,
\end{align}
where $a(Y)$ is the number of subsets $Z\subseteq Y$ that belong to $\cF=\cFd$.

For the number of $k$-partitions, we use an indeterminate $z$ in the ring $R$ that allows us to keep track of the sum of the cardinalities of the sets in the cover. The number of $k$-partitions for a subset $X\subseteq U$ is given by the coefficient of the monomial $z^{|X|}$ in the polynomial
\begin{align}\label{eq:partition}
    d_k(\cF,X) = \sum_{Y\subseteq X} (-1)^{|X\setminus Y|} \left( \sum_{j=0}^{|Y|} a_j(Y) z^j \right)^k,
\end{align}
where $a_j(Y)$ is the number of size-$j$ subsets $Z\subseteq Y$ that belong to \cF.

For the number of $k$-packings of $X$, we compute the number of $(k+1)$-partitions where the first $k$ members of the $(k+1)$-tuple belong to \cF and the last member is an arbitrary subset of $X$. Noting that $(1+z)^{|Y|} = \sum_{i=0}^{|Y|} \binom{|Y|}{i} z^i$, the number of $k$-packings for $X\subseteq U$ is the coefficient of $z^{|X|}$ in
\begin{align}\label{eq:packing}
    p_k(\cF,X) = \sum_{Y\subseteq X} (-1)^{|X\setminus Y|} (1+z)^{|Y|} \left( \sum_{j=0}^{|Y|} a_j(Y) z^j \right)^k.
\end{align}

\sloppy
The first algorithmic step is to compute the values for $a(Y)$ in \eqref{eq:cover} and the polynomials $\sum_{j=0}^{|Y|} a_j(Y) z^j$ in \eqref{eq:partition} and \eqref{eq:packing}.
Observe that $|Y|\le |X|\le \alpha n$.
To compute the values $a(Y)$ for all $\alpha$-small $Y\subseteq U$, observe that $a(Y) = \sum_{Z\subseteq Y} [Z\in \cF]$ is the $\alpha$-small zeta transform of the indicator function of \cF.
Since the implicit set system is polynomial time computable, the indicator function can be evaluated in polynomial time.
Therefore, the fast $\alpha$-small zeta transform allows us to compute all relevant values of $a(\cdot)$ in $O^*(s(n,\alpha))$ time.
Similarly, the polynomial $\sum_{j=0}^{|Y|} a_j(Y) z^j$ equals $h\zeta_\alpha(Y)$ where $h(Z) = [Z \in \cF] \cdot z^{|Z|}$ and can be computed in the same time bound by the fast $\alpha$-small zeta transform.

The second algorithmic step is to to use the fast $\alpha$-small M\"{o}bius transform and apply it to the functions that associate with each $\alpha$-small $Y\subseteq U$ the values $a(Y)^k$; $\left(\sum_{j=0}^{|Y|} a_j(Y) z^j\right)^k$; and $(1+z)^{|Y|}\left(\sum_{j=0}^{|Y|} a_j(Y) z^j\right)^k$, respectively.

We conclude that $c_k(\cF,X)$, $d_k(\cF,X)$, and $p_k(\cF,X)$ for each $\alpha$-small $X\subseteq U$ can be computed in $O^*(s(n,\alpha))$ time.

\begin{theorem}\label{thm:preprocessing}
    Given a polynomial time computable implicit set family $\Phi(I)=(U,\cF)$ with $|U|=n$, there is a $O^*(s(n,\alpha))$ time algorithm which determines, for all $k \le \alpha n$ and all $\alpha$-small $X\subseteq U$, whether $(X,\cF)$ has a $k$-cover (if we assume that $\cF$ is closed under subsets), $k$-partition, or $k$-packing.
\end{theorem} 


\section{Divide-and-conquer algorithm}
\label{sec:divide}

Let $\Phi(I)=(U,\cF)$ be a polynomial time computable implicit set family for which we would like to determine whether there is a $k$-cover (assuming $\cF=\cFd$), $k$-partition, or $k$-packing.

We say that a simple branching algorithm \emph{enumerates} a family $\cF'$ of subsets of $U$ if, at each leaf of its search tree it finds at most one member of $\cF'$, and collectively the leaves find all members of $\cF'$ (duplicates are allowed).

Our algorithm uses a divide-and-conquer strategy where the universe is twice partitioned into two.
\begin{definition}
    For a family of subsets \cF of a universe $U$, and a subset $X\subseteq U$,
    the \emph{restriction} of \cF to $X$ is $r(\cF,X) = \{S\subseteq X : S\in \cF\}$.
\end{definition}

Observe that
$(U,\cF)$ has a $k$-partition (resp., a $k$-packing or a $k$-cover) for $k\ge 2$ iff
there is a set $L\subseteq U$ and a positive integer $k_l<k$ such that
$(L,r(\cF,L))$ has a $k_l$-partition (resp., a $k_l$-packing or a $k_l$-cover) and $(R,r(\cF,R))$ has a $k_r$-partition (resp., a $k_r$-packing or a $k_r$-cover), where $R=U\setminus L$ and $k_r=k-k_l$.
We say that a $k$-partition, $k$-packing, or $k$-cover $(S_1,\dots,S_k)$ does not \emph{straddle} $X$ if for each $i\in \{1,\dots,k\}$, either $S_i\subseteq X$ or $S_i\cap X = \emptyset$.

In this section we prove the following theorem.

\begin{theorem}\label{thm:main-cover}
    Suppose there is a simple (classical) branching algorithm $A$, which, given an instance $I$ and a subset $X\subseteq U$,
    enumerates a family $e(X,I)$ of subsets of $X$ such that
    \begin{itemize}
        \item if $(U,\cF)$ has a $k$-cover (resp., a $k$-partition or a $k$-packing) that does not straddle $X$, then $(U,\cF)$ has a $k$-cover (resp., a $k$-partition or a $k$-packing) $(S_1,\dots,S_k)$ with $S_1\in e(X,I)$, and
        \item the algorithm runs in $O^*(c^{|X|})$ time for some $c\le 2$.
    \end{itemize}
    Then, there is a bounded-error quantum algorithm, which determines whether $(U,\cF)$ has a $k$-cover (resp., a $k$-partition or a $k$-packing) in $O^*\left( \binom{n}{n/4} + (2+c)^{n/2} \right)$ time, where $n=|U|$.
\end{theorem}

From now on, we focus on Set Cover; the discussion of Set Partition and Set Packing is analogous.
The first step is to use the algorithm from \autoref{thm:preprocessing} with $\alpha=\frac{1}{4}$ and store the result in QRAM.

Ideally, we would want to divide the universe $U$ into equal sized sets $L$ and $R$ and compute a $k$-cover where each set of the cover is responsible for covering elements of either $L$ or $R$.
However, we cannot guarantee that such a $k$-cover exists.
Instead, we can focus on partitions of $U$ into $L$ and $R$ with $|L|\ge n/2$ where the removal of one member of the $k$-cover decreases the size of $L$ to at most $n/2$.

\begin{lemma} \label{lma:ksplit}
    For any $x\in\{0,\dots,n\}$, $(U,\cF)$ has a $k$-cover, $k\ge 2$, iff
    there is a partition of $U$ into $(L,R)$ with $|L|\ge x$ and integers $k_L,k_R\ge 1$ with $k=k_L+k_R$ such that
    $(L,r(\cF,L))$ has a $k_L$-cover $(S_1,\dots,S_{k_L})$ and $(R,r(\cF,R))$ has a $k_r$-cover
    such that $|S_i| \ge |L|-x$ for every $i\in\{1,\dots,k_L\}$.
\end{lemma}
In particular, this means that removing any member $S_i$ of the $k_L$-cover gives a $(k_L-1)$-cover of $L\setminus S_i$ and $|L\setminus S_i|\le x$.

We use Grover's search to divide the elements into two sets $(L,R)$ where $|L|\ge n/2$ and $k$ into $k_L+k_R$.
Then, we solve each of these two instances independently, again using Grover's search on $(L,r(\cF,L))$ (resp., on $(R,r(\cF,R))$), dividing it into $(LL,LR)$ where $|LL|\ge n/4$ and $k_{LL}+k_{LR}=k_L$ (similar for $R$).
For $(LL,r(\cF,LL))$ (and, similarly, on the corresponding instances for $LR$, $RL$, $RR$), we use Grover's search on the branching algorithm $A$ to enumerate candidate sets $S_1\in e(LL,r(\cF,LL))$; if $|LL|-|S_1|>n/4$, then this branch is unsuccessful; otherwise, look up whether $LL\setminus S_1$ has a $(k_{LL}-1)$-cover in QRAM.

The running time of this algorithm is
\begin{align*}
    O^*\left(\binom{n}{n/4}\right)
\end{align*}
for the running the algorithm from \autoref{thm:preprocessing}; the steps using Grover's search take
\begin{align*}
    O^*\left(\sqrt{\sum_{l=\lceil n/2\rceil}^n \binom{n}{l} \sum_{l'=\lceil n/4 \rceil}^l \binom{l}{l'} c^{l'} }\right) = O^*\left((2+c)^{n/2}\right)
\end{align*}
time to achieve constant success probability. This proves \autoref{thm:main-cover}.

\paragraph{Discussion.}
When $c \ge \frac{16}{3^{3/2}}-2 \approx 1.079201$, then the running time is $O^*\left( (2+c)^{n/2} \right)$.
When $c \le \frac{16}{3^{3/2}}-2$, then the running time is dominated by the term $\binom{n}{n/4}\approx 1.7548^n$. 
For $c\le 1.147899$, the next section gives faster algorithms.


\section{Divide-and-conquer algorithms for small $c$}
\label{sec:smallc}
In this section, we again assume that there is an enumeration algorithm $A$, as in \autoref{thm:main-cover}, with running time $O^*\left(c^{|X|}\right)$. We present two divide-and-conquer algorithms which are faster for small values of $c$.

Throughout this section, we focus on Set Cover; the discussion of Set Partition and Set Packing is analogous.

\begin{theorem} \label{thm:smallc-cover}
    There is a bounded-error quantum algorithm, which determines whether $(U, \cF)$ has a $k$-cover (resp., a $k$-partition or a $k$-packing) in $O^*\left( \binom{n}{n/4} + (1+c)^{\frac{3}{4} \cdot n} \right)$ time, where $n=|U|$.
\end{theorem}
The first step is to use the algorithm from \autoref{thm:preprocessing} with $\alpha=\frac{1}{4}$ and store the result in QRAM. Our algorithm is similar to the algorithm in \autoref{thm:main-cover}, but we use the branching algorithm $A$ to remove a subset from $L$ before dividing it into $(LL, LR)$. Again, \autoref{lma:ksplit} is central to our algorithm.

We use Grover's Search to divide the elements into two sets $(L,R)$ where $|L| \ge n/2$ and $k$ into $k_L+k_R$. Then, we solve each of these two instances independently, again using Grover's search on $(R, r(\cF, R))$, dividing it into $(RL, RR)$ where $|RL| \ge n/4$ and $k_{RL} + k_{RR} = k_{R}$. For the instance $(L, r(\cF, L))$, we use Grover's search on the branching algorithm $A$ to enumerate candidate sets $S_1 \in e(L, r(\cF, L))$; if $|L| - |S_1| > n/2$ then this branch is unsuccessful; otherwise, we use Grover's search on the subinstance $(L \setminus S_1, r(\cF, L \setminus S_1))$, dividing it into $(LL, LR)$ where $|LL| \ge n/4$ and $k_{LL} + k_{LR} = k_{L} - 1$.

We process $(LL,r(\cF,LL))$ (and, similarly, the corresponding instances for $LR$, $RL$, $RR$) as we did in \autoref{thm:main-cover} -- we use Grover's search on the branching algorithm $A$ to enumerate candidate sets $S_1\in e(LL,r(\cF,LL))$; if $|LL|-|S_1|>n/4$, then this branch is unsuccessful; otherwise, look up whether $LL\setminus S_1$ has a $(k_{LL}-1)$-cover in QRAM.

The running time of this algorithm is
\begin{align*}
    O^*\left(\binom{n}{n/4}\right)
\end{align*}
for the running the algorithm from \autoref{thm:preprocessing}; the steps using Grover's search take
\begin{align*}
    O^*\left(\sqrt{\sum_{l=\lceil n/2\rceil}^n \binom{n}{l} c^{l} \sum_{l'=\lceil n/4 \rceil}^{\lfloor n / 2 \rfloor} \binom{n / 2}{l'} c^{l'} }\right) = O^*\left((1+c)^{\frac{3}{4} \cdot n}\right)
\end{align*}
time to achieve constant success probability. This proves \autoref{thm:smallc-cover}.

\paragraph{Discussion.}
When $c \ge \frac{2^{8/3}}{3} - 1 \approx 1.11653$, then the running time is $O^*\left( (1+c)^{\frac{3}{4} \cdot n} \right)$.
When $c \le \frac{2^{8/3}}{3} - 1$, then the running time is dominated by the term $\binom{n}{n/4}\approx 1.7548^n$. 
When $c \le 1.0872$, then the following algorithm is faster.

\begin{theorem} \label{thm:smallc-cover2}
    If $c \le 1.0872$, there is a bounded-error quantum algorithm, which determines whether $(U, \cF)$ has a $k$-cover (resp., a $k$-partition or a $k$-packing) in
    \begin{align*}
        O^*\left(\left(\min_{0.1303\le \alpha \le 0.25} \left((1 + c)^{3 / 4} c^{\alpha / 2}
        \left(1 - 4 \cdot \alpha\right)^{\frac{4 \cdot \alpha-1}{8}}
        \left(4 \cdot \alpha\right)^{-\frac{\alpha}{2}}
        ,\;\; \alpha^{-\alpha} \cdot (1-\alpha)^{\alpha-1}\right)\right)^n\right) \text{ time.}
    \end{align*}
\end{theorem}

The first step is to use the algorithm from \autoref{thm:preprocessing} with some $0.1303 \le \alpha < 0.25$ and store the result in QRAM.
Our algorithm is identical to the one presented in \autoref{thm:smallc-cover} except that we cannot directly look up whether a subset $X \subseteq U$ has a $k$-cover in QRAM for $|X| \le n / 4$. Instead, we use the following approach.

\begin{lemma} \label{lma:3rdsplit}
    If $c \le 1.0872$ and $0.1303 \le \alpha < 0.25$, after running the algorithm from \autoref{thm:preprocessing}, there exists a bounded-error quantum algorithm that checks for a $k$-cover of a subset $X \subseteq U$ where $|X| \le n/4$ in $O^*\left( \sqrt{\binom{n/4}{\alpha \cdot n} \cdot c^{\alpha \cdot n}} \right)$ time.
\end{lemma}
\begin{proof}
    We use Grover's search to divide the elements of $X$ into two sets $(XL, XR)$ where $|XL| \ge \alpha \cdot n$ and $k$ into $k_{XL} + k_{XR}$. We use Grover's search on the branching algorithm $A$ to enumerate candidate sets $S_1 \in e(XL, r(\cF, XL))$; if $|XL| - |S_1| > \alpha \cdot n$, then this branch is unsuccessful; otherwise, look up whether $XL \setminus S_1$ has a $(k_{XL} - 1)$-cover and $XR$ has a $k_{XR}$-cover in QRAM.
    The correctness of this approach follows from \autoref{lma:ksplit}.
    
    As the function $f(x) = \binom{n/4}{x} c^x$ is strictly decreasing for $x > \frac{c}{c + 1} \cdot \frac{n}{4}$ and 
    \begin{align*}
        \alpha \cdot n \ge 0.1303 \cdot n > \frac{1.0872 \mathbin{/} 4}{1.0872 + 1} \cdot n \ge \frac{c}{c + 1} \cdot \frac{n}{4},
    \end{align*}
    the running time of this algorithm is
    \begin{align*}
        O^*\left( \sqrt{\sum_{i = \lceil \alpha \cdot n \rceil}^{\lfloor n/4 \rfloor}\binom{n/4}{i} c^i} \right) = O^*\left( \sqrt{\binom{n/4}{\alpha \cdot n} c^{\alpha \cdot n}} \right)\\
    \end{align*}
\end{proof}

Our algorithm is the same as the one in \autoref{thm:preprocessing}, except when we check for a $(k_{LL}-1)$-cover for $LL \setminus S_1$ (resp., on $LR$, $RL$, $RR$) we use \autoref{lma:3rdsplit} instead of a direct lookup in QRAM.

Running the algorithm from \autoref{thm:preprocessing} takes
\begin{align*}
    O^*\left(\binom{n}{\alpha \cdot n}\right) &= O^*\left(\alpha^{-\alpha n} \cdot (1-\alpha)^{(\alpha-1)n}\right)
\end{align*}
time. The steps using Grover's search take
\begin{align*}
    &\; O^*\left(\sqrt{\sum_{l=\lceil n/2\rceil}^n \binom{n}{l} c^{l} \sum_{l'=\lceil n/4 \rceil}^{\lfloor n / 2 \rfloor} \binom{n / 2}{l'} c^{l'} } \sqrt{\binom{0.25 \cdot n}{\alpha \cdot n} c^{\alpha \cdot n}} \right)\\
    =&\; O^*\left(\left((1 + c)^{3 / 4} c^{\alpha / 2}
    \left(\frac{1}{1 - 4 \cdot \alpha}\right)^{\frac{1 - 4 \cdot \alpha}{8}}
    \left(\frac{1}{4 \cdot \alpha}\right)^{\frac{\alpha}{2}}
    \right)^n\right) \text{ time.}
\end{align*}

The running time of the first part increases with $\alpha$ while the running time of the second part decreases with $\alpha$. We can therefore optimise the running time by balancing the value of $\alpha$ for a given $c$. To determine when this approach is faster than $O^*\left(\binom{n}{n/4}\right)$ we compute the value of $c$ when $\alpha$ is balanced at $0.25$:
\begin{alignat*}{2}
    && (1 + c)^{3/4} c^{1 / 8} &= \frac{4}{3^{3/4}}\\
    &\;\:\Longrightarrow& c &\approx 1.08724\enspace.
\end{alignat*}

Therefore, this algorithm outperforms the algorithms of \autoref{thm:main-cover} and \autoref{thm:smallc-cover} when $c \le 1.08723$. The precise running time for certain values of $c$ and the corresponding values of $\alpha$ are shown in the table below. We also depict the best running times of all three algorithms for various values of $c$.

\begin{center}
    \begin{tabular}[t]{l l l} 
        \toprule
        $c$ & $\alpha$ & running time\\
        \midrule
        1.0 & 0.236159 & $O^*(1.7274^n)$ \\
        1.01 & 0.238036 & $O^*(1.7312^n)$ \\
        1.02 & 0.239858 & $O^*(1.7349^n)$ \\
        1.03 & 0.241622 & $O^*(1.7384^n)$ \\
        1.04 & 0.243320 & $O^*(1.7418^n)$ \\
        1.05 & 0.244946 & $O^*(1.7450^n)$ \\
        1.06 & 0.246488 & $O^*(1.7480^n)$ \\
        1.07 & 0.247928 & $O^*(1.7508^n)$ \\
        1.08 & 0.249227 & $O^*(1.7533^n)$ \\
        \bottomrule
    \end{tabular}
    \hfill
    \begin{tikzpicture}[baseline=(current bounding box.north)]
        \begin{axis}[
            title=Running time for various values of $c$,
            xlabel={$c$}, ylabel={Base $b$ of running time $O^*(b^n)$},
            ymin=1.7, ymax=2,
            xmin=1, xmax=2,
            legend style = {
                at = {(0.95,0.05)},
                anchor = south east,
                draw = none,
            },
            ]
            \addplot[violet, domain=1.1479:2.0, samples=100, thick,] {(2+x)^(0.5)}; \addlegendentry{\autoref{thm:main-cover}}
            \addplot[ red, domain=1.08724:1.1479, samples=60, thick,] {max(1.7548,(1+x)^(0.75))}; \addlegendentry{\autoref{thm:smallc-cover}}
            \addplot[blue, thick,] table {smallc.dat}; \addlegendentry{\autoref{thm:smallc-cover2}}
            
        \end{axis}
    \end{tikzpicture}
\end{center}


Combining \autoref{thm:main-cover}, \autoref{thm:smallc-cover}, and \autoref{thm:smallc-cover2}, we obtain the following corollary.

\begin{corollary}\label{cor:simpletime}
    There is a bounded-error quantum algorithm, which determines whether $(U, \cF)$ has a $k$-cover (resp., a $k$-partition or a $k$-packing) in $O^*\left( (2+c)^{n/2} \right)$ time, where $n=|U|$.
\end{corollary}

When $c \le 1.147899$, then the running time provided by the best among \autoref{thm:smallc-cover} and \autoref{thm:smallc-cover2} is slightly faster. See \autoref{sec:smallc-details} for details.
At $c = 1$, the running time of \autoref{thm:smallc-cover2} matches the running time of the Set Cover algorithm of \cite{AmbainisBIKPV19} when $|\mathcal{F}|$ is subexponential in $n$.

\section{Applications}
\label{sec:applications}

We can now use \autoref{cor:simpletime} in combination with simple branching algorithms enumerating various vertex sets in graphs, in particular algorithms enumerating
maximal independent sets (equivalently, maximal cliques in the complement graph) in $O^*(3^{n/3})$ time \cite{FominK10} (we note that the algorithm of \cite{JohnsonYP88} is not a \emph{simple} branching algorithm),
maximal bicliques in $O^*(3^{n/3})$ time,
maximal induced matchings in $O^*(10^{n/5})$ time \cite{GuptaRS12},
maximal induced forests in 
$O(1.8527^n)$ time \cite{GaspersL17}, and
minimal $k$-hitting sets in $O^*((2-1/k)^n)$ time \cite{FominGLS19}.
The last result is used to enumerate maximal $\mathcal{H}$-free subgraphs, which have no induced subgraph isomorphic to any graph from the family $\mathcal{H}$ of graphs, all of which have at most $k$ vertices.
Some well-known $\mathcal{H}$-free graph classes are
\begin{itemize}
    \item \emph{cluster} graphs with $\mathcal{H} = \set{P_3}$, where $P_k$ denotes the path on $k$ vertices,
    \item \emph{triangle-free} graphs with $\mathcal{H} = \set{K_3}$, where $K_k$ denotes the complete graph on $k$ vertices,
    \item \emph{cographs} with $\mathcal{H} = \set{P_4}$ \cite{CorneilLS81},
    \item \emph{claw-free} graphs with $\mathcal{H} = \set{K_{1,3}}$, where $K_{k,\ell}$ denotes the complete bipartite graph with partite sets of size $k$ and $\ell$,
    \item \emph{trivially-perfect} graphs with $\mathcal{H} = \set{P_4, C_4}$ \cite{Golumbic78}, where $C_k$ denotes the cycle on $k$ vertices,
    \item \emph{threshold} graphs with $\mathcal{H} = \set{P_4, C_4, \overline{C_4}}$ \cite{ChvatalH73}, where $\overline{G}$ denotes the complement of $G$,
    \item \emph{split} graphs with $\mathcal{H} = \set{C_4, \overline{C_4}, C_5}$ \cite{FoldesH77}, and
    \item \emph{line} graphs, where $\mathcal{H}$ is a set of 9 graphs on at most 6 vertices each \cite{Beineke70}.
\end{itemize}

\begin{theorem}
    There are bounded-error quantum algorithms, which, for a graph on $n$ vertices and integer $k$, determine whether there is a $k$-packing, $k$-partitioning, or $k$-covering of $G$
    \begin{itemize}
        \item with maximal independent sets, maximal cliques, or maximal bicliques in $O(1.8554^n)$ time,
        \item with maximal cluster graphs or maximal triangle-free graphs in $O(1.9149^n)$ time, 
        \item with maximal cographs, maximal claw-free graphs, maximal trivially-perfect graphs, or maximal threshold graphs in $O(1.9365^n)$ time, 
        \item with maximal split graphs in $O(1.9494^n)$ time,
        \item with maximal line graphs in $O(1.9579^n)$ time, and
        \item with maximal induced forests in $O(1.9629^n)$ time.
    \end{itemize}
    There are bounded-error quantum algorithms, which, for a graph on $n$ vertices and integer $k$, determine whether there is a $k$-packing or $k$-partitioning of $G$
    \begin{itemize}
        \item with maximal induced matchings in $O(1.8934^n)$ time\footnote{Note that it is NP-hard to determine, for a graph $G=(V,E)$ and vertex subset $X\subseteq V$, whether there is a superset $Y\supseteq X$ that induces a 1-regular subgraph of $G$. This can be seen by a simple reduction from Independent Set.}.
    \end{itemize}
\end{theorem}

In particular, this leads to a bounded-error quantum algorithm computing the chromatic number of an input graph in $O(1.8554^n)$ time; a graph can be covered with $k$ maximal independent sets iff it has chromatic number at most $k$.
In \autoref{sec:chromatic}, we expedite this algorithm by exploiting fast algorithms for coloring a graph with a small number of colors and this will enable us to partition the vertex set in a more balanced way in the divide-and-conquer steps.

Even though there is a simple branching algorithm that enumerates all minimal dominating sets in $O(1.7159^n)$ time \cite{FominGPS08}, this algorithm cannot be readily used for computing the domatic number using \autoref{thm:main-cover}. This is because when we consider a subset $X$ of the vertex set of $G=(V,E)$, we need to enumerate vertex subsets from $X$ that are minimal dominating sets for $G$, and not $G[X]$. In \autoref{sec:dom}, we prove that such minimal dominating sets can be enumerated in $O^*\left((2-\varepsilon)^{|X|}\right)$ time if $|X|$ is linear in $|V|$.

\section{Faster Computation of the Chromatic Number}
\label{sec:chromatic}

Assume the vertex set of input graph $G = (V, E)$ can be partitioned into independent sets $C = (I_1, I_2, \dots, I_\chi)$ where $\chi$ is the chromatic number of $G$. Denote by $n$ the number of vertices of $G$.

Using the algorithm from \autoref{thm:main-cover}, we can compute the chromatic number of $G$ in $O(1.8554^n)$ time. The family of subsets \cF corresponds to the independent sets of $G$ and the family $e(X, r(\cF,X))$ corresponds to the maximal independent sets of $G[X]$, which can be enumerated by a simple branching algorithm $A$ in $O(3^{|X|/3})$ time \cite{FominK10}.
We now prove a stronger result.

\begin{theorem}
    There is a bounded-error quantum algorithm, which for a graph on $n$ vertices, computes the chromatic number in $O(1.7956^n)$ time.
\end{theorem}

Assume, w.l.o.g., that $|I_1| \ge |I_2| \ge |I_3| \ge |I_4| \ge \dots \ge |I_\chi|$.

In the first step, we use the algorithm from \autoref{thm:preprocessing} with $\alpha = 0.27$ and store the result in QRAM. That is, for each $\alpha$-small subset $X \subset V$, we find the chromatic number of $X$ by finding the smallest $k$ such that there exists a $k$-partition of $X$ into independent sets. This takes $O^*(\binom{n}{0.27 \cdot n}) = O^*(1.79187^n)$ time.

We then consider a few different possibilities for the large sets in $C$ and present an algorithm which finds a partition into a smallest number of independent sets for each of these cases. These possibilities cover all configurations of $C$ so one result is guaranteed to find the chromatic number by detecting the partition $C$ above, or an equivalent partition. The cases are as follows:

\begin{enumerate}
    \item $G$ contains a vertex subset of size at least $0.48 \cdot n$ that is the union of at most five sets in $C$.\label{cs:5col}
    
    We check for $k$-coloring for all $1 \le k \le n$ and take the minimum valid value of $k$ as the chromatic number. To check for a certain $k$, we use Grover's search over all $Q \subset V$ where $|Q| \le 0.52 \cdot n$ and check whether $G-Q$ is 5-colorable. If it is 5-colorable, we compute its chromatic number by checking its $k'$-colorability for $k' < 5$. We then solve the instance $(Q, r(\cF, Q))$ to find a $k_Q$-coloring, where $k_Q = k - \chi(G-Q)$, using Grover's search to divide $Q$ into $(QL, QR)$ where $|QL| \ge 0.27 \cdot n$ and $k_{QL} + k_{QR} = k_Q$. For $(QL, r(\cF, QL))$ (and similarly for $QR$) we use Grover's search on the branching algorithm $A$ to enumerate candidate sets $S_1 \in e(QL, r(\cF,QL))$; if $|QL| - |S_1| > 0.27 \cdot n$, then this branch is unsuccessful; otherwise, find $\chi(QL \setminus S_1)$ in QRAM.
    
    The running time of this case with quantum 5-coloring in $O^*(1.4695^n)$ \cite{ShimizuM22} is
    \begin{align*}
        &\; O^*\left(
        \sum_{l = 0}^{\lfloor 0.52 \cdot n \rfloor} \sqrt{\binom{n}{l}} \left( \sum_{l' = \lceil \frac{l}{2} \rceil}^l \sqrt{\binom{l}{l'} 3^{l'/3}} + 1.4695^{n - l}\right) \right)\\
        =&\; O^*\left(
        \sqrt{\sum_{l = 0}^{\lfloor 0.52 \cdot n \rfloor} \binom{n}{l} \sum_{l' = \lceil \frac{l}{2} \rceil}^l \binom{l}{l'} 3^{l'/3}}
        + \sqrt{\sum_{l = \lceil 0.48 \cdot n \rceil}^{n} \binom{n}{l} \left( 1.4695^l \right)^2} \right)\\
        =&\; O^*\left(
        \sqrt{\binom{n}{0.52 \cdot n} 2.4423^{0.52 \cdot n}}
        + \sqrt{(1 + 1.4695^2)^n} \right)\\
        =&\; O^*\left(1.7831^n + 1.7775^n\right)
        =O^*(1.7831^n)
    \end{align*}
    
    \item $G$ contains a vertex subset of size at least $0.48 \cdot n$ that is the union of six sets in $C$, but does not contain a vertex subset of size at least $0.48 \cdot n$ that is the union of at most five sets in $C$. \label{cs:6col}
    
    In this case we have
    \begin{alignat*}{2}
        && |I_1| + |I_2| + \dots + |I_5| &< 0.48 \cdot n\\
        &\;\:\Longrightarrow& 5 \cdot |I_6| \le 5 \cdot |I_5| &< 0.48 \cdot n\\
        &\;\:\Longrightarrow& \;|I_1| + |I_2| + \dots + |I_6| &< 0.576 \cdot n,
    \end{alignat*}
    which, along with the condition $0.48 \cdot n \le |\bigcup_{i = 1}^6 I_i|$ gives
    \begin{align*}
        0.424 \cdot n < |V \setminus \bigcup_{i = 1}^6 I_i| \le 0.52 \cdot n.
    \end{align*}
    
    For this case we check for $k$-coloring for all $1 \le k \le n$ and take the minimum valid value of $k$ as the chromatic number. To check for a certain $k$, we use Grover's search over all $Q \subset V$ where $0.424 \cdot n < |Q| \le 0.52 \cdot n$ and check whether $G-Q$ is 6-colorable. If so, we check whether the instance $(Q, r(\cF, Q))$ is $(k - 6)$-colorable in the same way as in \autoref{cs:5col}.
    
    The running time of this case with quantum 6-coloring in $O^*(1.5261^n)$ \cite{ShimizuM22} is
    
    \begin{align*}
        &\; O^*\left(\sum_{l = \lceil 0.424 \cdot n \rceil}^{\lfloor 0.52 \cdot n \rfloor} \sqrt{\binom{n}{l}} \left( \sum_{l' = \lceil \frac{l}{2} \rceil}^l \sqrt{\binom{l}{l'} 3^{l'/3}} + 1.5261^{n - l}\right)
        \right)\\
        =&\; O^*\left(
        \sqrt{\sum_{l = \lceil 0.424 \cdot n \rceil}^{\lfloor 0.52 \cdot n \rfloor} \binom{n}{l} \sum_{l' = \lceil \frac{l}{2} \rceil}^l \binom{l}{l'} 3^{l'/3}}
        + \sqrt{\sum_{l = \lceil 0.48 \cdot n \rceil}^{\lfloor 0.576 \cdot n \rfloor} \binom{n}{l} \left( 1.5261^l \right)^2}
        \right)\\
        =&\; O^*\left(
        \sqrt{\binom{n}{0.52 \cdot n} 2.4423^{0.52 \cdot n}}
        + \sqrt{\binom{n}{0.576 \cdot n} 1.5261^{2\cdot 0.576 \cdot n}}
        \right)\\
        =&\; O^*\left(1.7831^n + 1.7937^n \right)
        = O^*(1.7937^n)
    \end{align*}
    
    \item $G$ does not contain a vertex subset of size at least $0.48 \cdot n$ that is the union of at most six sets in $C$.
    That is $|\bigcup_{i=1}^6 I_i| < 0.48 \cdot n$.
    
    Let $T = \bigcup_{i = 1}^{q} I_i$ where $q$ is the maximum index such that $|T| < \frac{n}{2}$. As $|\bigcup_{i = 1}^6 I_i| < 0.48 \cdot n$, we have $q \ge 6$. We consider two possibilities for the size of $T$ and present an algorithm which finds a partition into a smallest number of independent sets for each case. As a valid coloring is computed in each case, the smallest partition gives the chromatic number if $|\bigcup_{i=1}^6 I_i| < 0.48 \cdot n$.
    
    \begin{enumerate}
        \item Consider the case where $|T| < \frac{6 \cdot n}{13}$: \makeatletter\def\@currentlabel{3.1}\makeatother \label{cs:smallT}
        
        Let $L = \bigcup_{i = 1}^{q + 1} I_i$ and $R = \bigcup_{q + 2}^{\chi} I_i$. We have
        \begin{alignat*}{2}
            && |I_1| + |I_2| + \dots + |I_{q}| &< \frac{6 \cdot n}{13}\\
            &\;\:\Longrightarrow& 6 \cdot |I_{q+1}| \le q \cdot |I_{q}| &< \frac{6 \cdot n}{13}\\
            &\;\:\Longrightarrow& \;|I_1| + |I_2| + \dots + |I_{q+1}| &< \frac{7 \cdot n}{13}\\
            &\iff& |L|&< \frac{7 \cdot n}{13}
        \end{alignat*}
        and $0.5 \cdot n \le |L|$ from the definition of $q$.
        
        Assume $L$ is not 7-colorable. There must be more than 7 independent set in the construction of of $L$, which means
        \begin{alignat*}{2}
            && q + 1 > 7 \iff q \ge 7\\
            &\;\:\Longrightarrow& 7 \cdot |I_{q+1}| < \frac{6 \cdot n}{13}\\
            &\;\:\Longrightarrow& |L| < \frac{48 \cdot n}{91}
        \end{alignat*}
        
        By contraposition, $|L| \ge \frac{48 \cdot n}{91}$ implies that $L$ is 7-colorable.
        
        \begin{lemma} \label{lma:7color}
            After running the algorithm presented in \autoref{thm:preprocessing}, there exists a bounded-error quantum algorithm to check the 7-colorability of a subset $X \subseteq U$ where $|X| \le \frac{7}{3} \cdot \alpha \cdot n$ in $O^*(1.5622^{|X|})$ time.
        \end{lemma}
        
        \begin{proof}
            Assume that $X$ is 7-colorable and there is a partition $D = (J_1, J_2, \dots, J_7)$ of $X$ into independent sets. Assume, w.l.o.g., that $|J_1| \ge |J_2| \ge \dots \ge |J_7|$. Let $TL = J_1 \cup J_2 \cup J_3$ and $TR = J_4 \cup J_5 \cup J_6 \cup J_7$. We have
            \begin{alignat*}{3}
                && \frac{|TL|}{3} \ge |J_3| \ge |J_4| \ge \frac{|TR|}{4}\\
                &\;\:\Longrightarrow& |TL| + \frac{4}{3} \cdot |TL| \ge |X|\\
                &\iff& |TL| \ge \frac{3}{7} \cdot |X|.
            \end{alignat*}
            
            Consider $TR \setminus J_4$,
            \begin{alignat*}{2}
                && |TR| \le \frac{4}{7} \cdot |X| \text{ and }
                |J_4| \ge \frac{1}{4} \cdot |TR|\\
                &\;\:\Longrightarrow& |TR \setminus J_4| \le \frac{3}{7} \cdot |X| \le \alpha \cdot n.
            \end{alignat*}
            So there exists a subset $S_1 \in e(TR, r(\cF,TR))$ such that $|TR \setminus S_1| \le \alpha \cdot n$.
            
            We use Grover's search to divide $X$ into two sets $(XL, XR)$ where $|XL| \ge \frac{3 \cdot |X|}{7}$. We check $XL$ for 3-colorability using a fast quantum algorithm. If it is 3-colorable, we use Grover's search on the branching algorithm $A$ to enumerate candidate sets $S_1 \in e(XR, r(\cF,XR))$; if $|XR| - |S_1| > \alpha \cdot n$, then this branch is unsuccessful; otherwise, check whether $\chi(XR \setminus S_1) \le 3$ in QRAM. If $X$ is 7-colorable and $(XL, XR) = (TL, TR)$, we indeed detect $7$-colorability.
            
            The running time of this algorithm with quantum 3-coloring in $O^*(1.1528^n)$ time \cite{Furer08} is
            \begin{align*}
                &\; O^*\left(
                \sum_{l = \lceil \frac{3 \cdot |X|}{7} \rceil}^{|X|} \sqrt{\binom{|X|}{l}} \left( 1.1528^l + \sqrt{3^{(|X| - l) / 3}} \right)
                \right)\\
                =&\; O^*\left( \sqrt{\sum_{l = \lceil \frac{3 \cdot |X|}{7} \rceil}^{|X|} \binom{|X|}{l} \left(1.1528^l\right)^2 }
                + \sqrt{\sum_{l = 0}^{\lfloor \frac{4 \cdot |X|}{7} \rfloor} \binom{|X|}{l} 1.4423^l}
                \right)\\
                =&\; O^*\left( \sqrt{(1 + 1.1528^2)^{|X|}}
                + \sqrt{\binom{|X|}{\frac{4 \cdot |X|}{7}} 1.4423 ^ {\frac{4 \cdot |X|}{7}}}
                \right)\\
                =&\; O^*(1.5261^{|X|} + 1.5622^{|X|})
                = O^*(1.5622^{|X|})
            \end{align*}
        \end{proof}
        
        We check for $k$-coloring for all $1 \le k \le n$ and take the minimum valid value of $k$ as the chromatic number. To check for a certain $k$, we use Grover's search to divide $V$ into two sets $(L, R)$ where $0.5 \cdot n \le |L| < \frac{7 \cdot n}{13}$ and $k$ into $k_L + k_R$. When $|L| < \frac{48 \cdot n}{91}$, we solve the subinstances $(L, r(\cF, L))$ and $(R, r(\cF, R))$ as in \autoref{thm:main-cover}. When $|L| \ge \frac{48 \cdot n}{91}$, we solve the subinstance $(R, r(\cF, R))$ as in \autoref{thm:main-cover} and run the algorithm in \autoref{lma:7color} on $L$; if $L$ isn't 7-colorable, then this branch is unsuccessful; otherwise the subinstance $(L, r(\cF, L))$ is $k_L$-colorable for $k_L \ge 7$. This algorithm finds the chromatic number of $G$ when $(L, R)$ are equal to the ones we constructed from $C$.
        
        The running time of this algorithm is
        \begin{align*}
            &\; O^*\left(\sum_{l = \lceil 0.5 \cdot n \rceil}^{\lfloor \frac{7}{13} \cdot n \rfloor} \sqrt{\binom{n}{l}} \left(
            \begin{cases}
                \sum_{l' = \lceil \frac{l}{2} \rceil}^l \sqrt{\binom{l}{l'} 3^{l'/3}} &,\; l < \frac{48 \cdot n}{91}\\
                1.5622^l &,\; l \ge \frac{48 \cdot n}{91}
            \end{cases}
            + \sum_{l' = \lceil \frac{n - l}{2} \rceil}^{n - l} \sqrt{\binom{n - l}{l'} 3^{l'/3}}
            \right)\right)\\
            =&\; O^*\left(
            \sqrt{\sum_{l = \lceil \frac{6}{13} \cdot n \rceil}^{\lfloor \frac{48}{91} \cdot n \rfloor} \binom{n}{l} \sum_{l' = \lceil \frac{l}{2} \rceil}^l \binom{l}{l'} 3^{l'/3}}
            + \sqrt{\sum_{l = \lceil \frac{48}{91} \cdot n \rceil}^{\lfloor \frac{7}{13} \cdot n \rfloor} \binom{n}{l} \left(1.5622^l\right)^2}
            \right)\\
            =&\; O^*\left(
            \sqrt{\binom{n}{\frac{48}{91} \cdot n} 2.4423^{\frac{48}{91} \cdot n}}
            + \sqrt{\binom{n}{\frac{7}{13} \cdot n} 1.5622^{2\cdot \frac{7}{13} \cdot n}}
            \right)\\
            =&\; O^*\left(1.7884^n + 1.7956^n \right)
            = O^*(1.7956^n)
        \end{align*}
        
        \item Otherwise, consider the case where $\frac{6 \cdot n}{13} \le |T| < \frac{n}{2}$: \makeatletter\def\@currentlabel{3.2}\makeatother \label{cs:bigT}
        
        Let $L = T = \bigcup_{i = 1}^{q} I_i$ and $R = \bigcup_{q + 1}^{\chi} I_i$. We partition $R$ into independent sets $C' = (I'_1, I'_2, \dots, I'_{\chi'})$ where $I'_i = I_{q + i}$ for all $1 \le i \le \chi'$. By the definition of $C$, this is an optimal partition of $R$. As $q \ge 6$, we get $\frac{|L|}{6} \ge |I_{q}| \ge |I_{q + i}| = |I'_i|$ for all $1 \le i \le \chi'$.
        
        \begin{lemma} \label{lma:bound}
            Assume we have a partition $D = (J_1, J_2, \dots, J_{m})$ of a set $X$ where $t \ge |J_1| \ge |J_2| \ge \dots \ge |J_{m}|$. Define $p = \left\lceil \frac{|X|}{t} \right\rceil$ and $r = \frac{|X|}{p}$. For any non-negative integer $a$, $1 \le a \le p - 2$, there exists an integer $k$ such that $a \cdot r \le |\bigcup_{i = 1}^{k} J_i| \le (a + 1) \cdot r$.
        \end{lemma}
        
        \begin{proof}
            Let $k$ be the largest index such that $|\bigcup_{i = 1}^{k} J_i| \le (a + 1) \cdot r$. We will prove that $a \cdot r \le |\bigcup_{i = 1}^{k} J_i|$ and hence $k$ satisfies the constraints.\\
            
            Assume that $|\bigcup_{i = 1}^{k} J_i| < a \cdot r$. There are two cases:
            \begin{enumerate}
                \item[1.] $k \ge a$. We have 
                \begin{alignat*}{2}
                    && |J_1| + |J_2| + \dots + |J_k| &< a \cdot r\\
                    &\;\:\Longrightarrow& \;a \cdot |J_{k + 1}| \le k \cdot |J_{k}| &< a \cdot r\\
                    &\;\:\Longrightarrow& \;|J_1| + |J_2| + \dots + |J_{k+1}| &< a \cdot r + r
                \end{alignat*}
                which reaches a contradiction: $k$ is not the largest index such that $|\bigcup_{i = 1}^{k} J_i| \le (a + 1) \cdot r$.
                
                \item[2.] $k < a \implies k + 1 \le a$. We have
                \begin{alignat*}{2}
                    && |\bigcup_{i = 1}^{k + 1} J_i| &> (a + 1) \cdot r\\
                    &\;\:\Longrightarrow& \;(k + 1) \cdot t &> (a + 1) \cdot r\\
                    &\;\:\Longrightarrow& a \cdot t &> (a + 1) \cdot r.
                \end{alignat*}
                From $p = \left\lceil \frac{|X|}{t} \right\rceil$,
                \begin{alignat*}{2}
                    && p - 1 &< \frac{|X|}{t}\\
                    &\iff& \;(p - 1) \cdot t &< |X|\\
                    &\iff& (p - 1) \cdot t &< p \cdot r.
                \end{alignat*}
                We also have
                \begin{alignat*}{2}
                    && a < p - 1 \text{ and } t = \frac{|X|}{|X| \mathbin{/} t} \ge \frac{|X|}{\lceil |X| \mathbin{/} t \rceil} = r\\
                    &\;\:\Longrightarrow& (a - (p - 1)) \cdot t \le (a - (p - 1)) \cdot r.
                \end{alignat*}
                Adding the last two results, we get
                \begin{align*}
                    a \cdot t &< (a + 1) \cdot r
                \end{align*}
                which is a contradiction.
            \end{enumerate}
            Therefore $a \cdot r \le |\bigcup_{i = 1}^{k} J_i|$ and the proof is complete.
        \end{proof}
        
        Let $T_a = \bigcup_{i = 1}^{q'} I'_i$ and $T_b = \bigcup_{i = 1}^{q' + 1} I'_i$ where $q'$ is the maximum index such that $|T_a| < \frac{|R|}{2}$. Note that $\frac{|R|}{2} \le |T_b|$. When we apply \autoref{lma:bound} with $X = R$, $D = C'$ and $t = \frac{|L|}{6}$, we get
        \begin{alignat*}{2}
            && \frac{|X|}{t} = \frac{6 \cdot |R|}{|L|} \le \frac{6 \cdot \frac{7 \cdot n}{13}}{\frac{6 \cdot n}{13}} &= 7\\
            && \text{and\quad}\frac{6 \cdot |R|}{|L|} > \frac{6 \cdot \frac{n}{2}}{\frac{n}{2}} &= 6 \\
            &\;\:\Longrightarrow& p = \left\lceil \frac{|X|}{t} \right\rceil &= 7.
        \end{alignat*}
        If we let $a = 3$, the lemma states that there exists a $k$ such that $\frac{3 \cdot |R|}{7} \le |\bigcup_{i = 1}^{k} I'_i| \le \frac{4 \cdot |R|}{7}$. As $q'$ and $q' + 1$ differ by 1, it is not possible that $|\bigcup_{i = 1}^{q'} I'_i| < \frac{3 \cdot |R|}{7} \le |\bigcup_{i = 1}^{k} I'_i| \le \frac{4 \cdot |R|}{7} < |\bigcup_{i = 1}^{q' + 1} I'_i|$. So, either $\frac{3 \cdot |R|}{7} \le |T_a| < \frac{|R|}{2}$ or $\frac{|R|}{2} \le |T_b| \le \frac{4 \cdot |R|}{7}$.
        We let $TL = T_a$ or $TL = T_b$ such that $\frac{3 \cdot |R|}{7} \le |TL| \le \frac{4 \cdot |R|}{7}$ and $TR = R \setminus TL$.
        Note that in both cases, removing the independent set $I'_{q' + 1}$ from either $TL$ or $TR$ would leave both subsets with a size $\le \frac{|R|}{2}$.
        Therefore, there exists a subset $S_1 \in e(TL, r(\cF,TL))$ (resp. for $TR$) such that $|TL \setminus S_1| \le \frac{|R|}{2} \le 0.27$ (and similarly for $TR$).\\
        
        We check for $k$-coloring for all $1 \le k \le n$ and take the minimum valid value of $k$ as the chromatic number. To check for a certain $k$, we use Grover's search to divide $V$ into two sets $(L, R)$ where $\frac{6 \cdot n}{13} \le |L| < 0.5 \cdot n$ and $k$ into $k_L + k_R$. We solve the subinstance $(L, r(\cF, L))$ as in \autoref{thm:main-cover}. We also solve the subinstance $(R, r(\cF, R))$ as in \autoref{thm:main-cover}, except that $R$ is divided into $(RL, RR)$ with $\frac{3 \cdot |R|}{7} \le |RL| \le \frac{4 \cdot |R|}{7}$ instead of $|RL| > 0.5 \cdot |R|$. This algorithm finds the chromatic number of $G$ when $(L, R)$ are equal to the ones we constructed from $C$ and $(RL, RR) = (TL, TR)$.
        
        The running time of this algorithm is
        \begin{align*}
            &\; O^*\left(\sum_{l = \lceil \frac{6}{13} \cdot n \rceil}^{\lfloor 0.5 \cdot n \rfloor} \sqrt{\binom{n}{l}} \left(
            \sum_{l' = \lceil \frac{l}{2} \rceil}^l \sqrt{\binom{l}{l'} 3^{l'/3}}
            + \sum_{l' = \lceil \frac{3 \cdot (n - l)}{7} \rceil}^{\lfloor \frac{4 \cdot (n - l)}{7} \rfloor} \sqrt{\binom{n - l}{l'} 3^{l'/3}}
            \right)\right)\\
            =&\; O^*\left(
            \sqrt{\sum_{l = \lceil \frac{6}{13} \cdot n \rceil}^{\lfloor 0.5 \cdot n \rfloor} \binom{n}{l} \sum_{l' = \lceil \frac{l}{2} \rceil}^l \binom{l}{l'} 3^{l'/3}}
            + \sqrt{\sum_{l = \lceil 0.5 \cdot n \rceil}^{\lfloor \frac{7}{13} \cdot n \rfloor} \binom{n}{l} \sum_{l' = \lceil \frac{3 \cdot l}{7} \rceil}^{\lfloor \frac{4 \cdot l}{7} \rfloor} \binom{l}{l'} 3^{l'/3}}
            \right)\\
            =&\; O^*\left(
            \sqrt{\sum_{l = \lceil \frac{6}{13} \cdot n \rceil}^{\lfloor 0.5 \cdot n \rfloor} \binom{n}{l} 2.4423^l}
            + \sqrt{\sum_{l = \lceil 0.5 \cdot n \rceil}^{\lfloor \frac{7}{13} \cdot n \rfloor} \binom{n}{l} 2.4405^l}
            \right)\\
            =&\; O^*\left(
            \sqrt{\binom{n}{0.5 \cdot n} 2.4423^{0.5 \cdot n}}
            + \sqrt{\binom{n}{\frac{7}{13} \cdot n} 2.4405^{\frac{7}{13} \cdot n}}
            \right)\\
            =&\; O^*\left(1.7680^n + 1.7956^n \right)
            = O^*(1.7956^n)
        \end{align*}
    \end{enumerate}
\end{enumerate}

We observe that the overall running time of the algorithm is
\begin{align*}
    &\; O^*\left(\sqrt{
        \binom{n}{\frac{7}{13} \cdot n} \binom{\frac{7}{13} \cdot n}{\frac{4}{13} \cdot n} 3^{\frac{1}{3} \cdot \frac{4}{13} \cdot n}
    }\right)\\
    =&\; O^*\left(\left(
    \frac{\sqrt{13}}{2^{7/13} \cdot 3^{23/78}}
    \right)^n\right)\\
    =&\; O^*(1.7956^n)
\end{align*}
We reach this worst case when $\chi(G) = 13$ and $|I_1| = |I_2| = |I_3| = \dots = |I_{13}|$. An example of such a graph is the disjoint union of $\frac{n}{13}$ complete graphs on 13 vertices. If the universe is partitioned into 2, one part must have at least 7 independent sets which is then partitioned into two parts where one part has at least 4 independent sets. Hence we cannot improve on the running time with a different twice-partitioning strategy.

We also note that the current best known quantum algorithms for checking 13-colorability, 7-colorability and 4-colorability \cite{ShimizuM22} do not improve our running time. When iterating through $S_1 \in (X, r(\cF, X))$, we only need to consider $S_1$ with $|S_1| \ge |X| - \alpha \cdot n$ so we can use an improved upper bound \cite{Byskov04} when we only need to consider $|S_1| > |X| \mathbin{/} 3$. However, this is only the case when $|X| > \frac{3}{2} \cdot \alpha \cdot n$; so it does not affect our overall running time as $|X| = \frac{4}{13} \cdot n$ in the worst case.


The pseudocode for the algorithm in this section can be found in \autoref{sec:pseudocode}.

\section{Enumeration of Minimal Subset Dominating Sets}
\label{sec:dom}

In this section, we prove that the number of minimal dominating sets of a graph that are subsets of some linear-sized subset of vertices $X$ is at most $O^*\left((2-\varepsilon)^{|X|}\right)$. Moreover, they can be enumerated by a simple branching algorithm whose running time is within a polynomial factor of this bound.

\begin{theorem}\label{thm:mds}
    There is a simple branching algorithm, which, given any graph $G=(V,E)$ and any subset of vertices $X\subseteq V$ with $|X|\ge d\cdot |V|$ for some $d>0$,
    enumerates all minimal dominating sets of $G$ that are subsets of $X$ in $O^*\left((2-\varepsilon_d)^{|X|}\right)$ time, for some $\varepsilon_d>0$.
\end{theorem}

The theorem will follow from a slightly more general theorem about minimal set covers of a set system.
From a graph $G=(V,E)$ and a vertex subset $X\subseteq V$, we obtain a set system $(U,\cF)$ where $U=V$ and a set $S_x\in \cF$ for each $x\in X$ that contains the closed neighborhood of vertex $x$ in the graph $G$: $S_x = N_G[x]$.
Then, there is a 1-to-1 correspondence between inclusion-wise minimal dominating sets in $G$ that are subsets of $X$ and inclusion-wise minimal set covers of $(U,\cF)$.

\begin{theorem}\label{thm:sc}
    There is a simple branching algorithm, which, given any set system $(U,\cF)$ with $|U|\le r\cdot |\cF|$ for some $r>0$,
    enumerates all minimal set covers of $(U,\cF)$ in $O^*\left(2^{(1-\varepsilon_r)\cdot |\cF|}\right)$ time, for some $\varepsilon_r>0$.
\end{theorem}

\cite{FominGPS08} proved \autoref{thm:mds} for $X=V$ and \autoref{thm:sc} for $r=1$. In particular, their algorithm enumerates all minimal dominating sets of a graph on $n$ vertices in $O(1.7159^n)$ time.

\begin{proof}
    Let $r>0$.
    Let $\varepsilon = \frac{3(r+1)-\log (2^{3(r+1)}-1)}{3(r+1)^2}$.
    Consider the measure
    \begin{align*}
        \mu(U,\cF) = (1-(r+1)\varepsilon) |\cF| + \varepsilon |U|
    \end{align*}
    which associates a weight of $1-(r+1)\varepsilon>0$ to each set and a weight of $\varepsilon>0$ to each element of a set system $(U,\cF)$.
    We will show that every set system $(U,\cF)$ has at most $2^{\mu(U,\cF)}$ minimal set covers by induction on $|\cF|$.
    For a set system with $|U|\le r\cdot |\cF|$, this number is $2^{(1-(r+1)\varepsilon) |\cF| + \varepsilon |U|} \le  2^{(1-\varepsilon)\cdot |\cF|}$.
    The proof can easily be turned into a simple branching algorithm enumerating all minimal set covers whose search tree is the induction tree of this proof.
    
    The statement trivially holds when $|\cF|=0$ since such an instance has at most $1$ minimal set cover.
    For the induction, we consider three cases.
    \begin{enumerate}
        \item $\cF$ contains a set $S$ with $|S|\ge 3(r+1)$. Any minimal set cover either contains $S$ or not. Those that do not contain $S$ are also minimal set covers of $(U,\cF\setminus \{S\})$ and those that contain $S$ are made up of $S$ and a minimal set cover of $(U\setminus S, \cF\setminus \{S\})$. For the induction, we would like that
        \begin{align*}
            &&2^{\mu(U,\cF\setminus \{S\})} + 2^{\mu(U\setminus S, \cF\setminus \{S\})} &\le 2^{\mu(U, \cF)}\\
            &\iff&2^{\mu(U,\cF)-(1-(r+1)\varepsilon)} + 2^{\mu(U, \cF)-(1-(r+1)\varepsilon)-|S|\cdot \varepsilon} &\le 2^{\mu(U, \cF)}\\
            &\iff&2^{-1+(r+1)\varepsilon} + 2^{-1+(r+1)\varepsilon-|S|\cdot \varepsilon} &\le 1\\
            &\DOTSB\;\Longleftarrow\;&2^{-1+(r+1)\varepsilon} + 2^{-1-2(r+1)\varepsilon} &\le 1,
        \end{align*}
        and this inequality holds because $\varepsilon \le \frac{1}{r+1} \log\left(\frac{1+\sqrt{5}}{2}\right)$.
        
        \item There is an element $u\in U$ with frequency at most $3(r+1)$, i.e., $u$ occurs in at most $3(r+1)$ sets of \cF. Denote the sets that contain $u$ by $\mathcal{S} = \{S\in \cF : u\in S\}$. Since $u$ needs to be covered, each set cover contains at least one set from $\mathcal{S}$. We use induction on all $2^{|\mathcal{S}|}-1$ choices of including at least one set from $\mathcal{S}$ into the set covers and excluding the remaining sets from $\mathcal{S}$; each such choice leads to a set cover instance where we remove all sets in $\mathcal{S}$, and we remove all elements covered by the sets that are included in the set covers. Each such choice reduces the measure $\mu$ by more than $|\mathcal{S}|\cdot (1-(r+1)\varepsilon)$. Now, we would therefore like that
        \begin{alignat*}{2}
            &&(2^{|\mathcal{S}|} -1)\cdot 2^{\mu(U,\cF)-|\mathcal{S}|\cdot (1-(r+1)\varepsilon)} &\le 2^{\mu(U,\cF)}\\
            &\iff&(2^{|\mathcal{S}|} -1)\cdot 2^{-|\mathcal{S}|\cdot (1-(r+1)\varepsilon)} &\le 1\\
            &\iff&2^{-|\mathcal{S}|\cdot (1-(r+1)\varepsilon)} &\le (2^{|\mathcal{S}|} -1)^{-1}\\
            &\iff&-|\mathcal{S}|\cdot (1-(r+1)\varepsilon) &\le -\log(2^{|\mathcal{S}|} -1)\\
            &\iff&(r+1)\varepsilon-1 &\le \frac{-\log(2^{|\mathcal{S}|} -1)}{|\mathcal{S}|}\\
            &\iff&\varepsilon &\le \frac{|\mathcal{S}|-\log(2^{|\mathcal{S}|} -1)}{(r+1) |\mathcal{S}|}
        \end{alignat*}
        Note that $\frac{|\mathcal{S}|-\log(2^{|\mathcal{S}|} -1)}{(r+1) |\mathcal{S}|}$ decreases when $|\mathcal{S}|$ increases, and for the maximum possible value of $|\mathcal{S}|$, which is $3(r+1)$, the inequality holds with equality for the value of $\varepsilon$ given in the beginning of the proof.
        
        \item It remains to consider the case where all sets have size less than $3(r+1)$ and all elements have frequency more than $3(r+1)$. Since the sum of set sizes equals the sum of element frequencies, we have that $|\cF| \ge |U|$. Here, we use the result of Fomin et al. \cite{FominGPS08} who proved that the number of minimal set covers is at most $1.7159^{|\cF|}$. For the induction, we would like that
        \begin{alignat*}{2}
            &&1.7159^{|\cF|} &\le 2^{\mu(U,\cF)}\\
            &\iff& 2^{|\cF| \log 1.7159} &\le 2^{(1-(r+1)\varepsilon) |\cF| + \varepsilon |U|}\\
            &\DOTSB\;\Longleftarrow\;&0.779 |\cF| &\le (1-(r+1)\varepsilon) |\cF|\\
            &\iff& \varepsilon \le \frac{0.221}{r+1}
        \end{alignat*}
        and this inequality holds for the value of $\varepsilon$ given in the beginning of the proof.
    \end{enumerate}
    This concludes the proof of the theorem.
\end{proof}

\autoref{thm:main-cover} now lets us conclude the following.

\begin{corollary}
    There is a bounded-error quantum algorithm which computes the domatic number of any graph on $n$ vertices in $O((2-\varepsilon)^n)$ time for some constant $\varepsilon>0$.
\end{corollary}

\pagebreak
\pagebreak
\appendix
\section{A closer look at the running time for small $c$}
\label{sec:smallc-details}

\begin{figure}[H]
    \begin{tikzpicture}[baseline=(current bounding box.north)]
        \begin{axis}[
            title=Running time for various values of $c$,
            xlabel={$c$}, ylabel={Base $b$ of running time $O^*(b^n)$},
            ymin=1.725, ymax=1.775,yticklabel style={/pgf/number format/.cd,fixed,fixed zerofill,precision=3,},
            xmin=1, xmax=1.15,
            legend style = {
                at = {(0.95,0.05)},
                anchor = south east,
                draw = none,
            },width=\textwidth,
            height=\textwidth
            ]
            \addplot[purple, domain=1:2.0, samples=100, thick,] {(2+x)^(0.5)}; \addlegendentry{$(2+c)^{n/2}$} 
            \addplot[ red, domain=1.08724:1.1479, samples=80, thick,] {max(1.7548,(1+x)^(0.75))}; \addlegendentry{\autoref{thm:smallc-cover}}
            \addplot[blue, thick,] table {smallc.dat}; \addlegendentry{\autoref{thm:smallc-cover2}}
        \end{axis}
    \end{tikzpicture}
    \caption{\label{fig:smallcdet} Visual presentation of running times of \autoref{thm:smallc-cover} and \autoref{thm:smallc-cover2}.}
\end{figure}
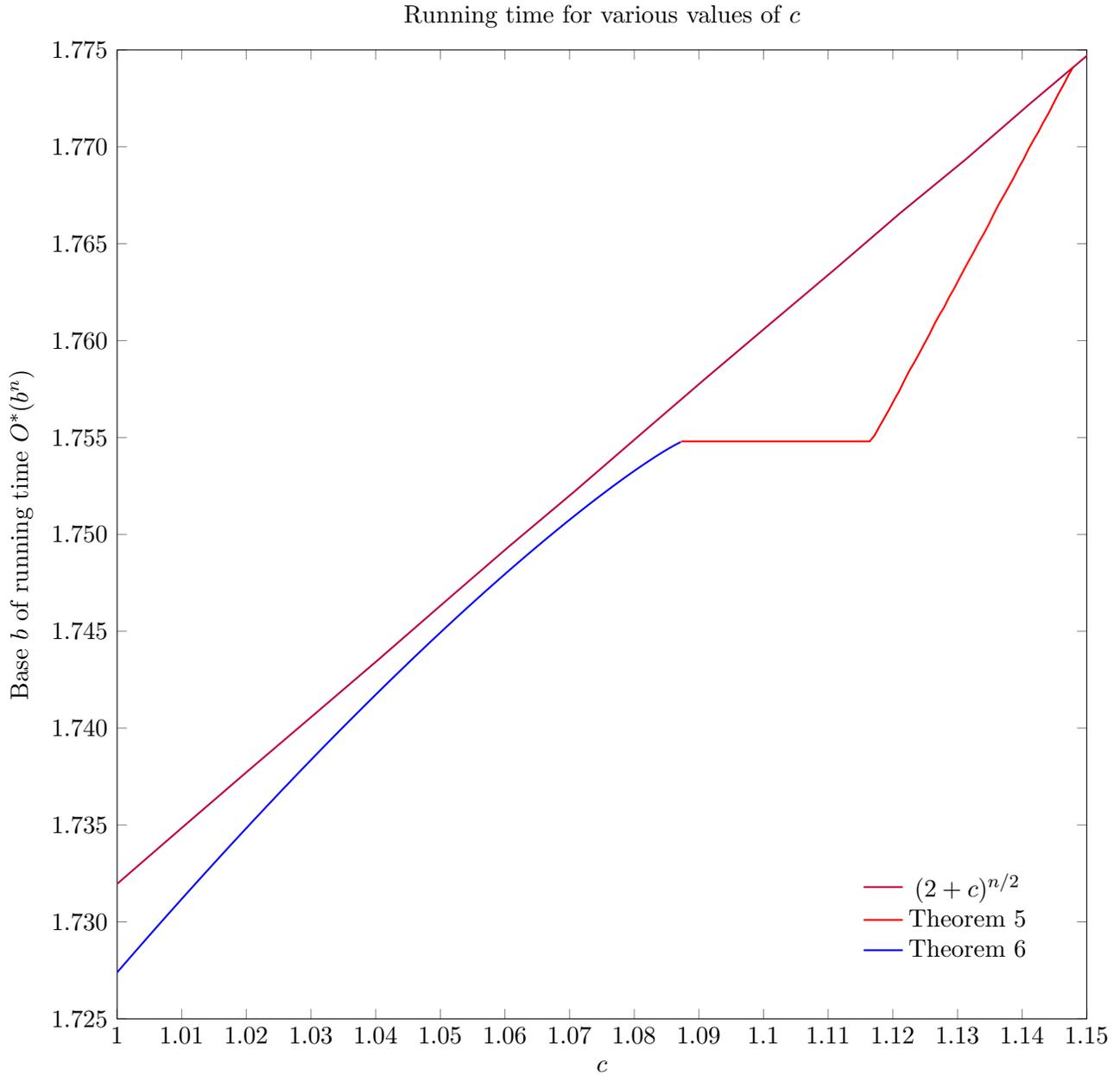

When $c < -\frac{2}{3} + \sqrt[3]{\frac{47}{54}-\frac{\sqrt{93}}{18}} + \sqrt[3]{\frac{47}{54}+\frac{\sqrt{93}}{18}} \approx 1.147899$, then choosing the best algorithm among \autoref{thm:smallc-cover} and \autoref{thm:smallc-cover2} gives an algorithm running in $O^*(b^n)$ time for some $b< \sqrt{c+2}$.
This can be seen visually in \autoref{fig:smallcdet} and proved rigorously, e.g., by interleaving a stepwise function between the function $(c+2)^{n/2}$ and the function from \autoref{thm:smallc-cover2} when $c\le 1.08$:
\begin{align*}
    \mathsf{steps}(c,n) &=
    \begin{cases}
        3^{n/2} & \text{if } c\le 1.0123\\
        1.735595^n & \text{if } 1.0123 \le c\le 1.0221\\
        \dots\\
        1.7533^n & \text{if } 1.0742 \le c\le 1.08.
    \end{cases}
\end{align*}
When $1.08 \le c \le 1.147899$, then \autoref{thm:smallc-cover} gives an algorithm with running time $O^*(d^n)$ for some $d<\sqrt{c+2}$.

\section{Pseudocode}
\label{sec:pseudocode}

\begin{algorithm}[h]
    \caption{\label{algo:precomp}Preprocessing for computing the chromatic number of $G$}
    \begin{algorithmic}[1]
        \Function{Precomp}{$G=(V,E)$}
        \State $a [\emptyset], \chi [\emptyset] \gets 0$
        \For {$S \subseteq V$, $|S| \le \lfloor 0.27 \cdot n \rfloor$} ($S$ is earlier than $S'$ if $S \subset S'$)
        \If {$S \neq \emptyset$}
        \State $v \gets $ arbitrary vertex of $S$
        \State $a[S] \gets a[S \setminus \{v\}] + a[S \setminus N[v]] + 1$ \Comment{Number of independent sets in $S$}
        \EndIf
        \EndFor
        \For {$k = \lfloor 0.27 \cdot n \rfloor, \lfloor 0.27 \cdot n \rfloor - 1, \ldots, 1$} 
        \For {$S \subseteq V$, $|S| \le \lfloor 0.27 \cdot n \rfloor$}
        \State $g_0[S] \gets (-1)^{|S|} (a[S])^k$
        \EndFor
        \For {$i = 1, 2, \ldots, n$}
        \For {$S \subseteq V$, $|S| \le \lfloor 0.27 \cdot n \rfloor$}
        \If {$v_i \in S$}   
        \State $g_i[S] = g_{i - 1}[S] + g_{i - 1}[X \setminus \{v_i\}]$
        \Else
        \State $g_i[S] = g_{i - 1}[S]$
        \EndIf
        \EndFor
        \EndFor
        \For {$S \subseteq V$, $|S| \le \lfloor 0.27 \cdot n \rfloor$}
        \If {$(-1)^{|S|} g_n[S] > 0$} \Comment{Number of $k$-covers of $S$}
        \State $\chi [S] \gets k$ \Comment{Chromatic number of $S$}
        \EndIf
        \EndFor
        \EndFor
        \EndFunction
    \end{algorithmic}
\end{algorithm}

In this section, we present the pseudocode for Chromatic number algorithm of \autoref{sec:chromatic}. As checking for a $k$-partitioning of independent subsets is equivalent to checking for a $k$-covering of maximal independent subsets, we can use the algorithm for computing $c_k(\cF,X)$ in \autoref{thm:preprocessing}. As the downward closure of the set of maximal independent subsets are independent subsets, we can adapt the preprocessing to directly compute the number of independent subsets of a set $X$.

In \autoref{algo:precomp}, the input is the graph $G=(V,E)$ with vertex set $V=\set{v_1,\dots,v_n}$.

We now present the main portion of the algorithm in \autoref{algo:main}. Grover's search is used to quadratically speed up the branching and is applied to every for loop. Here, instead of iteratively updating $c$, all values $1 \le c \le n$ are considered and $c$ is set to the minimum possible. In \textsc{Col}, we use fast quantum $k$-coloring algorithms presented in \cite{Furer08,ShimizuM22}.

\begin{algorithm}[h]
    \caption{\label{algo:main} Main algorithm for computing the chromatic number of $G$}
    \begin{algorithmic}[1]
        \Function{Main}{$G$}
        \State \textsc{Precomp}$(G)$
        \State $c \gets |V|$
        \For{$S \subseteq V$, $0.48 \cdot n \le |S|$} \Comment{Case \ref{cs:5col}}
        \If {\textsc{Col}$(5, \,S)$}
        \State $k \gets \min_{x \in \set{1, 2, 3, 4, 5}} \set{\textsc{Col}(x, \,S)}$
        \State $c \gets \min \{c,\, k + \textsc{Chr1a}(V \setminus S)\}$
        \EndIf
        \EndFor
        \For{$S \subseteq V$, $0.48 \cdot n \le |S| < 0.576 \cdot n$} \Comment{Case \ref{cs:6col}}
        \If {\textsc{Col}$(6, \,S)$}
        \State $c \gets \min \{c,\, 6 + \textsc{Chr1a}(V \setminus S)\}$
        \EndIf
        \EndFor
        \For{$S \subseteq V$, $0.5 \cdot n \le |S| < \frac{7 \cdot n}{13}$} \Comment{Case \ref{cs:smallT}}
        \If {$|S| < \frac{48 \cdot n}{91}$}
        \State $c \gets \min \{c,\, \textsc{Chr1a}(S) + \textsc{Chr1a}(V \setminus S)\}$
        \Else
        \State $c \gets \min \{c,\, \textsc{7Col}(S) + \textsc{Chr1a}(V \setminus S)\}$
        \EndIf
        \EndFor
        \For{$S \subseteq V$, $\frac{6 \cdot n}{13} \le |S| < 0.5$} \Comment{Case \ref{cs:bigT}}
        \State $c \gets \min \{c,\, \textsc{Chr1a}(S) + \textsc{Chr1b}(V \setminus S)\}$
        \EndFor
        \State \Return $c$
        \EndFunction
        \\
        \Function{Col}{$k$, $S$}
        \State \Return $k$-colorability of $S$ using \cite{Furer08,ShimizuM22}
        \EndFunction
    \end{algorithmic}
\end{algorithm}

\pagebreak
We now present the implementation for \textsc{Chr1a}, \textsc{Chr1b} and \textsc{7Col} in \autoref{algo:chr1}. 
Similarly to above, Grover's search is used to quadratically speed up the branching and is applied to every for loop. In \textsc{Chr2}, Grover's search is applied to a classical branching algorithm that enumerates the maximal independent sets of $G[S]$ in $O^*(3^{1/3})$ \cite{FominK10}.

\begin{algorithm}[h]
    \caption{\label{algo:chr1} Implementations of \textsc{Chr1a}, \textsc{Chr1b} and \textsc{7Col}}
    \begin{algorithmic}[1]
        \Function{Chr1a}{$S$}
        \State $c \gets |S|$
        \For{$T \subseteq S$, $\frac{|S|}{2} \le |T|$}
        \State $c \gets \min \{c,\textsc{Chr2}(T) + \textsc{Chr2}(S \setminus T)\}$
        \EndFor
        \State \Return $c$
        \EndFunction
        \\
        \Function{Chr1b}{$S$} \Comment{Algorithm used for $(R, r(\cF, R))$ in Case \ref{cs:bigT}}
        \State $c \gets |S|$
        \For{$T \subseteq S$, $\frac{3 \cdot |S|}{7} \le |T| \le \frac{|4| \cdot |S|}{7}$}
        \State $c \gets \min \{c,\textsc{Chr2}(T) + \textsc{Chr2}(S \setminus T)\}$
        \EndFor
        \State \Return $c$
        \EndFunction
        \\
        \Function{7Col}{$S$} \Comment{Algorithm presented in \autoref{lma:7color}}
        \For{$T \subseteq S$, $\frac{3 \cdot |S|}{7} \le |T|$}
        \If{\textsc{Col}$(3, \,T)$ \textbf{and} \textsc{Chr2}$(S \setminus T) \le 4$}
        \State \Return $7$
        \EndIf
        \EndFor
        \State \Return $|S|$
        \EndFunction
        \\
        \Function{Chr2}{$S$}
        \State $c \gets |S|$
        \For{$T \in \textsc{Mis}(G[S])$, $|S \setminus T| \le \lfloor 0.27 \cdot n \rfloor$}
        \State $c \gets \min \{c, 1 + \chi [S \setminus T]\}$
        \EndFor
        \State \Return $c$
        \EndFunction
    \end{algorithmic}
\end{algorithm}

\pagebreak

\bibliography{references}
\bibliographystyle{alphaurl}

\end{document}